\documentclass[10pt,twocolumn,aps,english,pra,superscriptaddress,longbibliography]{revtex4-2} 

\usepackage[utf8]{inputenc}

\usepackage{bibunits}
\usepackage{tikz} 
\usepackage{amsmath,amsfonts,amsthm, mathtools}
\usepackage{placeins}
\usepackage{float}
\usepackage{dsfont}
\usepackage{csquotes}
\usepackage{amssymb}
\usepackage{verbatim}
\usepackage{bm}
\usepackage{amsmath}
\usepackage{amssymb}
\usepackage{stmaryrd}
\usepackage{amsthm}
\usepackage[caption=false]{subfig}
\usepackage{physics}
\usepackage{bm}  

\usepackage{hyperref}
\definecolor{darkred}  {rgb}{0.5,0,0}
\definecolor{darkblue} {rgb}{0,0,0.5}
\definecolor{darkgreen}{rgb}{0,0.5,0}
\hypersetup{
  colorlinks = true,
  urlcolor  = blue,         
  linkcolor = darkblue,     
  citecolor = darkgreen,    
  filecolor = darkred       
}

\theoremstyle{definition}

\newtheorem{proposition}{Proposition}
\newtheorem{thm}{Theorem}

\usepackage{multirow}

\definecolor{cool_green}{rgb}{0.0, 0.5, 0.0}

\begin{document}

\title{Gaussian quantum data hiding}

\author{Yunkai Wang}
\email{y697wang@uwaterloo.ca}
\affiliation{Institute for Quantum Computing, University of Waterloo, Ontario N2L 3G1, Canada.}
\affiliation{Department of Applied Mathematics, University of Waterloo, Ontario N2L 3G1, Canada.}
\affiliation{Perimeter Institute for Theoretical Physics, Waterloo, Ontario N2L 2Y5, Canada.}

\author{Graeme Smith}
\email{graeme.smith@uwaterloo.ca}
\affiliation{Institute for Quantum Computing, University of Waterloo, Ontario N2L 3G1, Canada.}
\affiliation{Department of Applied Mathematics, University of Waterloo, Ontario N2L 3G1, Canada.}

\begin{abstract}
Quantum data hiding encodes a hidden classical bit to a pair of quantum states that is difficult to distinguish using a particular set of measurement, denoted as $M$. In this work, we explore quantum data hiding in two contexts involving Gaussian operations or states. First, we consider the set of measurement $M$ as Gaussian local quantum operations and classical communication, a new set of operations not previously discussed in the literature for data hiding. We hide one classical bit in the two different mixture of displaced two-mode squeezed states. Second, we consider the set of measurement $M$ as general Gaussian measurement and construct the data hiding states using two-mode thermal states. This data hiding scheme is effective in the weak strength limit, providing a new example compared to existing discussions for the set of general Gaussian measurement.
\end{abstract}

\maketitle

\section{Introduction}

Discriminating between quantum states or channels is a fundamental task in quantum information science \cite{duan2009perfect,acin2001statistical,duan2007entanglement,salek2022usefulness,becker2021energy,duan2008local}. The fundamental lower bound of success probability of distinguishing the states is known as Helstrom bound \cite{holevo1973statistical,holevo1976investigations,helstrom1967detection,helstrom1976quantum}. The optimal measurement used to saturate the Helstrom bound is chosen from all physically allowed positive operator valued measure (POVM), which can be hard to implement. A practically interesting problem is whether we can limit the available measurement to a set $M$ and ask for the success probability using this set \cite{walgate2000local,virmani2001optimal,matthews2010entanglement,yu2012four,fan2004distinguishability,duan2009distinguishability,childs2013framework}. For example, in the multipartite setting, local quantum operations assisted by
classical communication (LOCC) is an interesting set $M$ \cite{chitambar2014everything}. 

Quantum data hiding conceals classical bits against a specific set of POVMs $M$ by encoding the bits into quantum states, ensuring that the probability of successfully distinguishing between these states is no better than random guessing when only POVMs from the set $M$ are employed. This phenomenon was initially identified in Ref.~\cite{terhal2001hiding,divincenzo2002quantum} for the set of LOCC. This discovery allows for the secure hiding of one bit of classical information from any attempts to cheat using LOCC. 
Since then, the concept of quantum data hiding has been extended in various ways. For instance, hiding classical bits in multipartite settings is discussed in Ref.~\cite{eggeling2002hiding}. Protocols for hiding quantum data in both bipartite and multipartite cases are constructed \cite{divincenzo2003hiding,hayden2004randomizing,hayden2005multiparty}. 
Data hiding in the presence of noise is considered \cite{lupo2016quantum}.
It is shown that using many copies of specific data hiding states does not provide a disproportionate advantage over using a single copy \cite{matthews2009chernoff}.
Additionally, distinguishability norms have been introduced to study different sets of POVMs \cite{matthews2009distinguishability}.

In this work, we focus on the continuous variable (CV) version of quantum data hiding. Unlike discrete-variable quantum information, which primarily relies on qubits as the fundamental units of information, CV quantum information adopts a continuous-variable approach. This framework utilizes quantum states of systems characterized by continuous degrees of freedom, such as position and momentum of light fields \cite{braunstein2005quantum,weedbrook2012gaussian,adesso2014continuous,holevo2011probabilistic}. The primary tools in CV quantum information processing are Gaussian states and Gaussian operations. Gaussian states are characterized by their representation through Gaussian functions, making them mathematically convenient and experimentally accessible. Gaussian operations, in turn, are transformations that map Gaussian states to other Gaussian states.
We will discuss two set of POVM: (1) Gaussian local quantum operations assisted by classical communication (GLOCC) (2) general Gaussian operations.  We note there exists some discussion of data hiding from Gaussian operations \cite{sabapathy2021bosonic,lami2021quantum}.

For the first setting, data hiding from GLOCC, Ref.~\cite{lami2021quantum} discusses a closely related result under different conditions. They examine data hiding using CV states from LOCC, which includes non-Gaussian LOCC in their considered set $M$. They allow both the state and measurement to be non-Gaussian, but lacks concrete examples. In contrast, we consider GLOCC as the set $M$ and provide a concrete example of data hiding state which is a mixture of Gaussian states.



For the second setting, data hiding from general Gaussian measurements, Ref.~\cite{lami2021quantum} provides an example using single-mode CV data hiding states known as even and odd thermal states. However, it is important to note that the state considered in their work is neither Gaussian nor a mixture of Gaussian states. Additionally, Ref.~\cite{sabapathy2021bosonic} demonstrated the existence of two Gaussian states with infinitely many modes, each being a mixture of a finite set of randomly chosen coherent states. In both examples, the two states cannot be distinguished by Gaussian measurements, while general measurements can achieve a success probability close to one. Their constructions rely on quite nonclassical states or random mixtures of coherent states, whereas our approach is based on simple thermal states, providing an example of  derandomized and concrete data hiding states. Furthermore, while their discussion is rooted in information-theoretic arguments, our work specifies the exact measurements required for the data hiding scheme.


\section{Data Hiding from GLOCC}

In the following sections, we will begin with an instructive discussion on the performance of information extraction using GLOCC measurements on displaced two-mode squeezed states.
Using this understanding, we will then construct data hiding states that are secure against GLOCC.

\subsection{GLOCC measurements perform less effectively in decoding information from two-mode squeezed states}

In this subsection, we aim to build some intuition about when nonlocal Gaussian measurements can outperform GLOCC measurements in decoding information. Assume Alice and Bob share a displaced two-mode squeezed state as described by
\begin{equation}\label{rho_two_mode}
\rho_{\vec{r}}=(D_{a+ib}\otimes D_{c+id})\ket{\phi}\bra{\phi} (D_{a+ib}\otimes D_{c+id})^\dagger,   
\end{equation}
where $\vec{r}=[a,b,c,d]^T$, $D_\alpha=\exp(\alpha \hat{a}_i^\dagger-\alpha^* \hat{a}_i)$ is the displacement operation, $\hat{a}_i$ represents the annihilation operator for the corresponding mode on Alice's or Bob's sides, $\ket{\phi}=\exp(s(\hat{a}_1\hat{a}_2-\hat{a}_1^\dagger\hat{a}_2^\dagger)/2)\ket{0}$ is the two mode squeezed state with squeezing parameter $s$. The parameters $a,b,c,d$ are encoded on the quadratures of $\rho_{\vec{r}}$. Assume these parameters follow a prior distribution described by
\begin{equation}\begin{aligned}
\label{eq:prior_I}
P(\vec{r})&=\frac{1}{(2\pi)^2\sqrt{\det V_r}}\exp[-\frac{1}{2}\vec{r}^TV_r^{-1}\vec{r}],\\
&=\frac{1}{4\pi^2\sigma^4}\exp{-(a^2+b^2+c^2+d^2)/2\sigma^2}. 
\end{aligned}\end{equation}
Alice and Bob want to measure $\rho_{\vec{r}}$ with the POVM $\{M_x\}_x$ to decode the parameters $a,b,c,d$. We can treat this problem as a communication model with input $a,b,c,d$ and output $x$. In the following, we want to show that the mutual information for a proper nonlocal Gaussian measurement is much greater than any GLOCC. 

Any Gaussian measurement can be written as the form \cite{adesso2014continuous,weedbrook2012gaussian}
\begin{equation}\label{Pi_y}
\Pi_{\vec{y}}=\frac{1}{\pi^2}D_{\vec{y} }\Pi_0D^\dagger_{\vec{y}},
\end{equation}
where $\Pi_0$ is a density matrix of a general  Gaussian state with vanishing displacement and covariance matrix $V_\Pi$.  Note that the label $\vec{y}$ of the outcome is only related to the displacement of $\Pi_{\vec{y}}$. We now want to find the probability distribution $P(\vec{y}|\vec{r})=\tr(\Pi_{\vec{y}}\rho_{\vec{r}})$, where $\vec{r}$ is the information we want to send. Note for two operators $A,B$ whose Wigner function is $W_{A}(\vec{q},\vec{p}),W_{B}(\vec{q},\vec{p})$, we have
$\tr[AB]\propto\int d\vec{q}d\vec{p}W_{A}(\vec{q},\vec{p})W_{B}(\vec{q},\vec{p})$ \cite{case2008wigner}. The Wigner function for $\Pi_{\vec{y}}$ and $\rho_{\vec{r}}$ is given by
\begin{equation}
W_\Pi(q_1,p_1,q_2,p_2)=\frac{1}{\pi^2}\frac{\exp[-\frac{1}{2}(\vec{x}-\vec{y})^TV_\Pi^{-1}(\vec{x}-\vec{y})]}{(2\pi)^2\sqrt{\det V_{\Pi}}},
\end{equation}
\begin{equation}
W_\rho(q_1,p_1,q_2,p_2)=\frac{\exp[-\frac{1}{2}(\vec{x}-\vec{r})^TV_\rho^{-1}(\vec{x}-\vec{r})]}{(2\pi)^2\sqrt{\det V_\rho}},
\end{equation}
where $\vec{x}=[q_1,p_1,q_2,p_2]^T$, $\vec{y}=[y_1,y_2,y_3,y_4]^T$, the covariance matrix of $\rho_{\vec{r}}$ is given by
\begin{equation}\label{eq:TMSV}
V_\rho=\left[\begin{matrix}
\cosh 2s I & \sinh 2s Z\\
\sinh 2s Z & \cosh 2s I
\end{matrix}\right].
\end{equation}
Since the integral is simply a Gaussian integral, we can easily find
\begin{equation}\begin{aligned}
&P(\vec{y}|\vec{r})=\frac{1}{(2\pi)^2\sqrt{\det V}}\exp[-\frac{1}{2}(\vec{y}-\vec{r})^TV^{-1}(\vec{y}-\vec{r})],
\end{aligned}\end{equation}
where $V=V_\Pi+V_\rho$. And the mutual information is
\begin{equation}
I(\vec{y};\vec{r})=\frac{1}{2}\log(\det(I+\sigma^2 V^{-1})),\quad V=V_\Pi+V_\rho.
\end{equation}
We now aim to optimize $V_\Pi$ to maximize the mutual information $I(\vec{y};\vec{r})$ in the case of nonlocal Gaussian measurement and GLOCC.

\begin{proposition}\label{prop:nonlocal_Gaussian}
There exists a nonlocal Gaussian measurement with outcome labeled by $\vec{y}$ in Eq.~\ref{Pi_y} on the two-mode squeezed state $\rho_{\vec{r}}$ given in Eq.~\ref{rho_two_mode} encoded with the parameter $\vec{r}$ that can achieve the mutual information scaling $I(\vec{y};\vec{r})=2s$ to the leading order as $s\rightarrow\infty$.
\end{proposition}
\begin{proof}
We note the eigenspectrum of $V_\rho$ are,
\begin{equation}\begin{aligned}\label{Vrho_spectrum}
&\lambda_{\rho1}=
\lambda_{\rho2}=e^{2s},\quad\lambda_{\rho3}=\lambda_{\rho4}=e^{-2s},\\
&\omega_{\rho1}=\frac{1}{\sqrt{2}}[0,-1,0,1]^T,\quad \omega_{\rho2}=\frac{1}{\sqrt{2}}[1,0,1,0]^T,\\
&\omega_{\rho3}=\frac{1}{\sqrt{2}}[0,1,0,1]^T,\quad \omega_{\rho4}=\frac{1}{\sqrt{2}}[-1,0,1,0]^T.\\
\end{aligned}\end{equation}
Intuitively, the homodyne detection and beam splitter can estimate linear combinations of all the quadratures $q_1=a+dq_1$, $p_1=b+dp_1$, $q_2=c+dq_2$, $p_2=d+dp_2$, where $a,b,c,d$ are the displacement, $dq_{1,2}, dp_{1,2}$ are the intrinsic noise depending on the quantum state. We emphasize that while homodyne detection enables perfect precision in estimating certain quadratures, this precision only implies that no additional noise is introduced by the measurement itself; the intrinsic noise inherent to the quantum state still persists. For the two mode squeezed state considered here, if $r\rightarrow\infty$, we know $dq_1-dq_2$ and $dp_1+dp_2$ are approaching zero. But for example, if we just focus on $dq_1$, it is totally random. Due to this property, we should only be able to estimate $a-c+dq_1-dq_2\rightarrow a-c$ and $b+d+dp_1+dp_2\rightarrow b+d$.  These are the linear combinations that can be extracted from the two-mode squeezed state with perfect precision at an infinite squeezing level. This corresponds to estimating $q_1-q_2$ and $p_1+p_2$ which is a nonlocal measurement implemented with beam splitter and homodyne detection. The corresponding $V_\Pi$ in the eigenbasis $\{\omega_{\rho i}\}_{i=1,2,3,4}$ of $V_\rho$ is simply 
\begin{equation}\label{nonlocal}
V_\Pi=\lim_{\delta\rightarrow0}\left[
\frac{1}{\delta}(\omega_{\rho1}\omega_{\rho1}^T+ \omega_{\rho2}\omega_{\rho2}^T)+{\delta}(\omega_{\rho3}\omega_{\rho3}^T+ \omega_{\rho4}\omega_{\rho4}^T)\right].
\end{equation}
We can find that for this $V_\Pi$, 
\begin{equation}
\det(I+\sigma^2 V^{-1})\propto e^{4s},\quad I(\vec{y};\vec{r})= 2s+o(s),
\end{equation}
to the leading order as $s\rightarrow\infty$ and $e^{-2s}\gg \delta$.
\end{proof}

\begin{thm}\label{thm:I_GLOCC}
Any GLOCC measurement with outcome labeled by $\vec{y}$ in Eq.~\ref{Pi_y} on the two-mode squeezed state $\rho_{\vec{r}}$ given in Eq.~\ref{rho_two_mode} encoded with the parameter $\vec{r}$ cannot achieve the mutual information scaling $I(\vec{y};\vec{r})=2s$ to the leading order as $s\rightarrow\infty$.
\end{thm}

\begin{proof}
We will prove any GLOCC cannot achieve $I(\vec{y};\vec{r})=2s$ to the leading order as $s\rightarrow\infty$ by contradiction.
Firstly, we want to prove that to have $I(\vec{y};\vec{r})= 2s+o(s)$, or $\det(I+\sigma^2 V^{-1})\propto e^{4s}$, $V_\Pi$ has exactly two  eigenvalues approach zero ($\ll e^{-2s}$). Because if $V_\Pi$ has less than two eigenvalues approach zero, based on Weyl's inequality for $n\times n$ matrices $A,B$ \cite{horn2012matrix}
\begin{equation}
\lambda_{i+j-1}(A+B)\leq \lambda_i(A)+\lambda_j(B)\leq \lambda_{i+j-n}(A+B),
\end{equation}
where $\lambda_1\geq \lambda_2\geq \cdots\geq \lambda_n$.
We have 
\begin{equation}\begin{aligned}
&\lambda_3(V_\Pi+V_\rho)\geq \lambda_3(V_\Pi)+\lambda_4(V_\rho)\geq \lambda_3(V_\Pi)\sim O(e^{0r}),\\
&\lambda_4(V_\Pi+V_\rho)\geq \lambda_4(V_\Pi)+\lambda_4(V_\rho)\geq \lambda_4(V_\rho)\sim O(e^{-2s}).
\end{aligned}\end{equation}
Then, $V=V_\Pi+V_\rho$ only has one small eigenvalues scaling as $e^{-2s}$, which shows $\det(I+\sigma^2 V^{-1})$ at most scales as $e^{2s}$. Furthermore, $V_\Pi$ cannot have more than two eigenvalues approaching zero. 
Because for any real symmetric matrix $V_\Pi$, the requirement for it to be the covariance matrix is given by \cite{weedbrook2012gaussian}
\begin{equation}
V_\Pi>0,\quad V_\Pi+i\Omega\geq 0,\quad \Omega=\left[\begin{matrix}
0 & 1 & 0 & 0\\
-1 & 0 & 0 & 0\\
0 & 0 & 0 & 1\\
0 & 0 & -1 & 0
\end{matrix}\right].
\end{equation}
If $\lambda_{2,3,4}(V_\Pi)\rightarrow0$, and we know the eigenvalues of $i\Omega$ are 1,1,-1,-1. We can find that 
\begin{equation}
\lambda_4(V_\Pi+i\Omega)\leq \lambda_2(V_\Pi)+\lambda_3(i\Omega)\sim 0-1<0,
\end{equation}
which violates the requirement that $V_\Pi+i\Omega\geq 0$. This actually shows that $\lambda_{1,2}(V_\Pi)\geq 1$.

Secondly, we aim to demonstrate that the eigenvectors of $V_\Pi$ must correspond to the eigenvectors of $V_\rho$, as specified below. If the eigenvector $\omega_V$ of $V=V_\Pi+V_\rho$ has the corresponding eigenvalue $\lambda_V\rightarrow0$, i.e. $\omega_V$ is one of $\omega_{V3,4}$
\begin{equation}
V\omega_V=\lambda_V\omega_V.
\end{equation}
We can multiply $\omega_V^T$ from the left on the both side of the equation, which gives
\begin{equation}
\omega_V^TV\omega_V=\lambda_V\rightarrow0.
\end{equation}
If we write the spectrum decomposition $V_\Pi=\sum_{i=1,2,3,4}\lambda_{\Pi i} \omega_{\Pi i}\omega_{\Pi i}^T$, $V_\rho=\sum_{i=1,2,3,4}\lambda_{\rho i}\omega_{\rho i}\omega_{\rho i}^T$, then we can find that
\begin{equation}\begin{aligned}
&e^{-2s}[(\omega^T_V\omega_{\rho 3})^2+(\omega^T_V\omega_{\rho 4})^2]+e^{2s}[(\omega^T_V\omega_{\rho 1})^2\\
&+(\omega^T_V\omega_{\rho 2})^2]+\sum_i\lambda_{\Pi i}(\omega_V^T\omega_{\Pi i})^2=\lambda_V\rightarrow0.
\end{aligned}\end{equation}
Then, it is clear that $\omega^T_V\omega_{\Pi 1,2}\rightarrow0 $, $\omega^T_V\omega_{\rho 1,2}\rightarrow0$. We must have $\omega_{\Pi 3,4},\omega_{V 3,4}$ in the $\text{span}\{\omega_{\rho 3}, \omega_{\rho 4}\}$, $\omega_{\Pi 1,2},\omega_{V 1,2}$ in the $\text{span}\{\omega_{\rho 1}, \omega_{\rho 2}\}$.

Finally, based on the observed requirements for $V_\Pi$ to achieve $\det(I+\sigma^2 V^{-1})\propto e^{4s}$—namely, $\lambda_{V3,4} \rightarrow 0$, with $\omega_{\Pi 3,4}$ confined to the $\text{span}\{\omega_{\rho 3}, \omega_{\rho 4}\}$ and $\omega_{\Pi 1,2}$ confined to the $\text{span}\{\omega_{\rho 1}, \omega_{\rho 2}\}$—we are now prepared to demonstrate that $V_\Pi$ cannot be GLOCC. 
This can be easily verified because if
\begin{equation}\begin{aligned}
&V_\Pi=\sum_{i=1,2,3,4}\lambda_{\Pi i}\omega_{\Pi i} \omega_{\Pi i}^\dagger,\quad \lambda_{\Pi 1,2}\geq 1,\quad \lambda_{\Pi 3,4}\ll 1,\\
&\omega_{\Pi 1}\!=\!x_1\omega_{\rho 1}\!+\!\sqrt{1-x_1^2}\omega_{\rho2},\, \omega_{\Pi2}\!=\!\sqrt{1-x_1^2}\omega_{\rho1}\!-\!x_1\omega_{\rho2},\\
&\omega_{\Pi 3}\!=\!x_3\omega_{\rho3}\!+\!\sqrt{1-x_3^2}\omega_{\rho4},\,\omega_{\Pi 4}\!=\!\sqrt{1-x_3^2}\omega_{\rho3}\!-\!x_3\omega_{\rho4},\\
\end{aligned}\end{equation}
where $\omega_{\rho i}$ are the eigenvector of $V_\rho$ given in Eq.~\ref{Vrho_spectrum}, $|x_{1,3}|\leq 1$. $V_\Pi$ follows the Positive Partial Transpose (PPT) criterion iff $TV_\Pi T+i\Omega\geq 0$ \cite{weedbrook2012gaussian}, where 
\begin{equation}
T=\left[\begin{matrix}
1 & 0 & 0 & 0\\
0 & 1 & 0 & 0\\
0 & 0 & 1 & 0\\
0 & 0 & 0 & -1
\end{matrix}\right],
\end{equation}
and $TV_\Pi T$ corresponds to partial transpose.
 We can find one of eigenvalues of $TV_\Pi T+i\Omega$ is
\begin{equation}
\frac{1}{2}\left(\lambda_{\Pi 3}+\lambda_{\Pi 4}-\sqrt{4+(\lambda_{\Pi 3}-\lambda_{\Pi 4})^2}\right).
\end{equation}
It is clear that when $\lambda_{\Pi 3,4}\ll 1$, this eigenvalue is smaller than 0, which implies $V_\Pi$ cannot be a GLOCC measurement. We have thus showed that GLOCC measurements can only have mutual information  worse than $2s$. 
\end{proof}

Although the analytical proof shows that GLOCC cannot achieve $I(\vec{y}; \vec{r}) = 2s$, we are unable to analytically determine a tight upper bound for $I(\vec{y}; \vec{r})$ using GLOCC.
We now aim to numerically optimize the GLOCC measurement to determine the optimal mutual information \( I(\vec{y}; \vec{r}) \). As illustrated in Fig.~\ref{fig:mutual_information}, for small squeezing parameters \( s \ll 1 \), optimal \( I(\vec{y}; \vec{r}) \) using GLOCC remains nearly constant. This behavior is expected, as in the absence of squeezing, the measurement effectively reduces to estimating the displacement of a coherent state. However, as the squeezing parameter \( s \) increases, the optimal mutual information \( I(\vec{y}; \vec{r})\approx s \) under GLOCC measurements. This numerical result aligns with and confirms the analytical findings in Theorem \ref{thm:I_GLOCC}, demonstrating the limitations of GLOCC. For nonlocal Gaussian measurement, $I(\vec{y}; \vec{r})\approx 2s $ as stated in Proposition~\ref{prop:nonlocal_Gaussian}.

\begin{figure}[!tb]
\begin{center}
\includegraphics[width=0.4\textwidth]{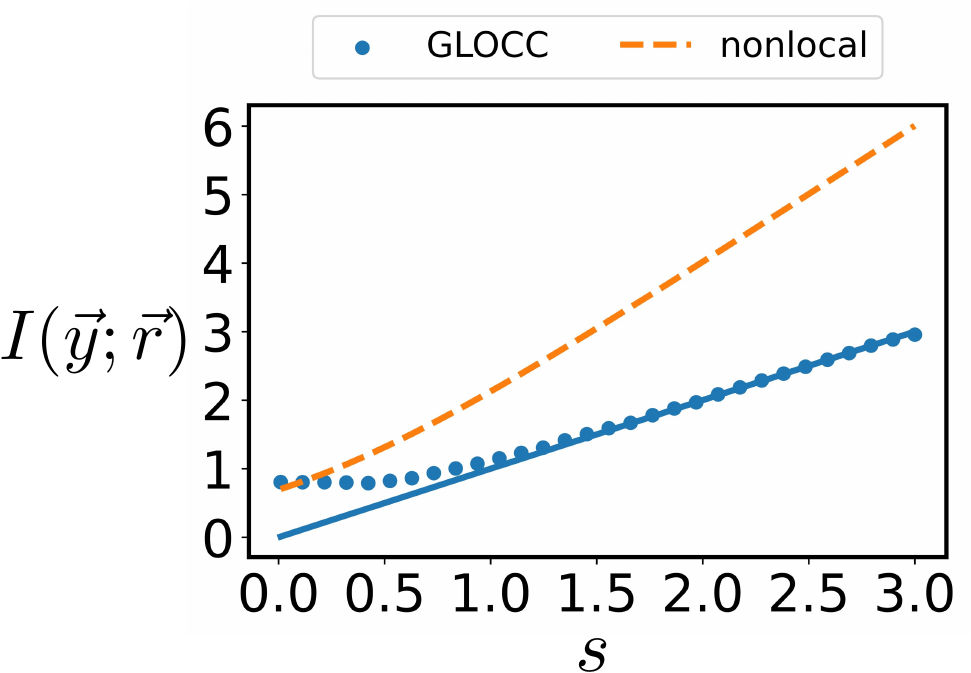}
\caption{The mutual information $I(\vec{y}; \vec{r})$ is presented as a function of the squeezing parameter $s$ for both GLOCC and nonlocal Gaussian measurements. For GLOCC, $I(\vec{y}; \vec{r})$ is obtained through numerical optimization. In the case of nonlocal Gaussian measurements, $I(\vec{y}; \vec{r})$ is calculated based on the measurements of $q_1 - q_2$ and $p_1 + p_2$, as introduced in Proposition \ref{prop:nonlocal_Gaussian}. The solid line represents the fitting function $I(\vec{y}; \vec{r}) = r$. Here, the parameter $\sigma$ is set to 1.} 
\label{fig:mutual_information}
\end{center}
\end{figure}

The above results serve as the CV counterpart to the fact that in the discrete-variable case, LOCC cannot fully distinguish the four Bell states. Local projection onto ${\ket{0}, \ket{1}}$ can only distinguish the states $\ket{00} \pm \ket{11}$ or $\ket{01} \pm \ket{10}$. Similarly, local projection onto ${\ket{+}, \ket{-}}$ can distinguish either $\ket{00} + \ket{11}$ and $\ket{01} + \ket{10}$, or $\ket{11} - \ket{00}$ and $\ket{01} - \ket{10}$. For Gaussian states, nonlocal Gaussian measurements can simultaneously estimate both $q_1 - q_2$ and $p_1 + p_2$ with perfect precision. In contrast, GLOCC can only estimate either $q_1 - q_2$, or $p_1 + p_2$, (or any linear combination of them), by relying on local homodyne detection. This intuition also leads to a no-go theorem: 
\begin{thm}
If we want to estimate any single linear combination $t_1(a - c) + t_2(b + d)$ in Eq.~\ref{rho_two_mode}, the performance of nonlocal Gaussian measurements and GLOCC measurements labeled by $\vec{y}$ as in Eq.~\ref{Pi_y} will be comparable in the sense that the mutual information for both cases are $I(\vec{y};\vec{r})=s$ to the leading order as $s\rightarrow\infty$.
\end{thm}
\begin{proof}
Intuitively, this is because a local homodyne with an appropriate phase shift can estimate any linear combination $k_1a + k_2b$ and $k_3c+k_4 d$, enabling us to find $t_1(a - c) + t_2(b + d)$. 
To support this intuition, we can choose
\begin{equation}\begin{aligned}\label{eq:linear_ts}
&P(\vec{r})=\frac{1}{(2\pi)^2\sqrt{\det V_r}}\exp[-\frac{1}{2}\vec{r}^TV_r^{-1}\vec{r}],\\
&V_r={\sigma^2}\omega_1\omega_1^T+\frac{1}{\delta^2}\sum_{i=2}^4\omega_i\omega_i^T,\quad\delta\rightarrow0,\\
&\omega_1=[t_1,t_2,-t_1,t_2]^T,\quad \omega_2=[t_1,t_2,t_1,-t_2]^T,\\
&\omega_3=[-t_2,t_1,t_2,t_1]^T,\quad \omega_4=[-t_2,t_1,-t_2,-t_1]^T,
\end{aligned}\end{equation}
where, without loss of generality, we choose $t_1^2 + t_2^2 = 1/2$ for normalization. We fix all other relationships between $a$, $b$, $c$, $d$ by setting $\delta \rightarrow 0$, ensuring that all the information is encoded in the linear combination $t_1(a - c) + t_2(b + d)$. We find
\begin{equation}\begin{aligned}
&I(\vec{y};\vec{r})=\frac{1}{2}\log(\det(I+\sigma^2 M V^{-1})),\\
&V=V_\Pi+V_\rho,\quad M=\left[\begin{matrix}
1 & 0 & 0 & 0\\
0 & 0 & 0 & 0\\
0 & 0 & 0 & 0\\
0 & 0 & 0 & 0\\
\end{matrix}\right],
\end{aligned}\end{equation}
where we write the matrix in the basis $\{w_i\}_{i=1,2,3,4}$ in Eq.~\ref{eq:linear_ts}. It is easy to verify that we can choose $V_\Pi$ for both both GLOCC measurements and nonlocal measurement such that we have 
\begin{equation}
V=\left[\begin{matrix}
e^{-2s} & 0 & 0 & 0\\
0 & * & 0 & 0\\
0 & 0 & * & 0\\
0 & 0 & 0 & *\\
\end{matrix}\right],
\end{equation}
in the basis $\{w_i\}_{i=1,2,3,4}$ in Eq.~\ref{eq:linear_ts}. For GLOCC measurement, we choose to estimate $t_1q_1+t_2p_1$ and $-t_1q_2+t_2p_2$. For nonlocal measurement, we can directly estimate both $q_1-q_2$ and $p_1+p_2$ using homodyne detection. So, both GLOCC measurements and nonlocal measurement have $I(\vec{y};\vec{r})=s+o(s)$.
    
\end{proof}

\subsection{Data hiding}

Building on the intuition from the previous subsection, we aim to construct a data hiding scheme. GLOCC cannot simultaneously obtain $a - c$ and $b + d$ with high precision, whereas a nonlocal measurement can achieve precise estimation of both parameters. We propose to use the sign of $\alpha = (a - c)(b + d)$ to encode one bit of classical information.
\begin{proposition}\label{GLOCC_one}
Consider the pair of data hiding states 
\begin{equation}
\rho_{+}=2\int_{\alpha>0}d\vec{r}\, P(\vec{r})\rho_{\vec{r}},\quad \rho_{-}=2\int_{\alpha<0}d\vec{r}\, P(\vec{r})\rho_{\vec{r}},
\end{equation}
where the prefactor 2 is introduced for normalization, $\alpha = (a - c)(b + d)$, $\rho_{\vec{r}}$ is given in Eq.~\ref{rho_two_mode}, $P(\vec{r})$ is given in Eq.~\ref{eq:prior_I}. The sign of $\alpha$ is used to encoded one bit of classical information. 
The sign of $\alpha$ can be determined using nonlocal Gaussian measurements, with the success probability approaching 100\%. However, using only GLOCC measurements, the probability of correctly determining the sign of $\alpha$ will sufficiently deviate from 100\%.
\end{proposition}
\begin{proof}

For two probability distribution $P_0(y)$ and $P_1(y)$ and outcome $\vec{y}$, the error probability of distinguishing the two distribution is given by \cite{csiszar2011information,pardo2018statistical}
\begin{equation}\begin{aligned}
&P_{\text{err}}=\frac{1}{2}(1-TV), \\
&TV=\frac{1}{2}||P_0-P_1||=\frac{1}{2}\sum_{y}|P_0(y)-P_1(y)|.  
\end{aligned}\end{equation}
where $TV$ is total variation distance between two probability distribution. Our aim is to maximize the total variation $TV$ by designing the POVM $\{\Pi_{\vec{y}}\}_{\vec{y}}$. The total variation between two probability distributions  $P(\vec{y}|\pm)=\tr(\Pi_{\vec{y}}\rho_{\pm})$
\begin{equation}\begin{aligned}
&TV=\frac{1}{2}\int d\vec{y}\left|P(\vec{y}|+)-P(\vec{y}|-)\right|\\
&=\int d\vec{y}\bigg|\int_{\alpha>0}d\vec{r}P(\vec{y}|\vec{r})P(\vec{r}) -\int_{\alpha<0}d\vec{r}P(\vec{y}|\vec{r})P(\vec{r})\bigg|\\
&=\int_{\beta>0}d\vec{y}\left(\int_{\alpha>0}d\vec{r}P(\vec{y}|\vec{r})P(\vec{r})-\int_{\alpha<0}d\vec{r}P(\vec{y}|\vec{r})P(\vec{r})\right)\\
&\quad+\int_{\beta<0}d\vec{y}\left(\int_{\alpha<0}d\vec{r}P(\vec{y}|\vec{r})P(\vec{r})-\int_{\alpha>0}d\vec{r}P(\vec{y}|\vec{r})P(\vec{r})\right)\\
&=TV(++)-TV(-+)+TV(--)-TV(+-),
\end{aligned}\end{equation}
where we   assume $\beta=(y_1-y_3)(y_2+y_4)$, $TV(\pm\pm)=TV(\alpha=\pm,\beta=\pm)$ represents the four different integrals.

To do the above calculation, let's define $\vec{Y}$ and $\vec{R}$ 
\begin{equation}\begin{aligned}\label{eq:U}
&U=\frac{1}{\sqrt{2}}\left[\begin{matrix}
1 & 0 & -1 & 0\\
1 & 0 & 1 & 0\\
0 & 1 & 0 & -1\\
0 & 1 & 0 & 1
\end{matrix}\right],\\
&\vec{Y}=U\vec{y},\quad \vec{R}=U\vec{r},\quad \Sigma=UVU^T.
\end{aligned}\end{equation}
And the integral can be written as
\begin{equation}
\begin{aligned}
&TV(++)=\left(\int_0^\infty\!dR_1\int_0^\infty\!dR_4+\int_{-\infty}^0\!dR_1\int_{-\infty}^0\!dR_4\right)\\
&\times\left(\int_0^\infty\!dY_1\int_0^\infty\!dY_4+\int_{-\infty}^0\!dY_1\int_{-\infty}^0\!dY_4\right)\\
&\times\int_{-\infty}^{+\infty}\!dR_2\int_{-\infty}^{+\infty}\!dR_3\int_{-\infty}^{+\infty}\!dY_2\int_{-\infty}^{+\infty}\!dY_3\,P(\vec{y}|\vec{r})P(\vec{r}).
\end{aligned}
\end{equation}
We  further define $\vec{P}=\vec{Y}+\vec{R}$, $\vec{Q}=\vec{Y}-\vec{R}$, and  have $P(\vec{y}|\vec{r})=\exp[-\frac{1}{2}\vec{Q}^T\Sigma^{-1}\vec{Q}]/\left[(2\pi)^2\sqrt{\det \Sigma}\right]$, $P(\vec{r})={\exp[-\frac{1}{8}(\vec{Q}-\vec{P})^TV_r^{-1}(\vec{Q}-\vec{P})]}/\left[(2\pi)^2\sqrt{\det V_r}\right]$.
Through the procedure of changing variables and changing the order the integration, we  do the integration for $\vec{P}$ first since we have a simple assumption that $V_r=\sigma^2 I$, which gives
\begin{equation}
\begin{aligned}
&TV(++)=\frac{1}{4}\int_{-\infty}^{+\infty}\!dQ_1\!\int_{-\infty}^{+\infty}\!dQ_2\!\int_{-\infty}^{+\infty}\!dQ_3\!\int_{-\infty}^{+\infty}\!dQ_4 \\
&\quad\times P(\vec{y}|\vec{r})\left(2-\text{erf}\left(\frac{|Q_1|}{\sqrt{2}\sigma}\right)-\text{erf}\left(\frac{|Q_4|}{\sqrt{2}\sigma}\right)\right)\\
&\quad+\frac{1}{2}(\int_{0}^{+\infty}\!dQ_1\!\int_{0}^{+\infty}\!dQ_4+\int_{-\infty}^0\!dQ_1\!\int_{-\infty}^0\!dQ_4)\\
&\quad\times\int_{-\infty}^{+\infty}\!dQ_2\!\int_{-\infty}^{+\infty}\!dQ_3\,P(\vec{y}|\vec{r})\text{erf}\left(\frac{|Q_1|}{\sqrt{2}\sigma}\right)\text{erf}\left(\frac{|Q_4|}{\sqrt{2}\sigma}\right).
\end{aligned}
\end{equation}
Similarly, we can find $TV(+-),TV(--),TV(-+)$ and eventually get
\begin{equation}\begin{aligned}
TV&=\int_{-\infty}^{+\infty}dQ_1\int_{-\infty}^{+\infty}dQ_2\int_{-\infty}^{+\infty}dQ_3\int_{-\infty}^{+\infty}dQ_4 \\
&\times P(\vec{y}|\vec{r})\left(1-\text{erf}\left(\frac{|Q_1|}{\sqrt{2}\sigma}\right)\right)\left(1-\text{erf}\left(\frac{|Q_4|}{\sqrt{2}\sigma}\right)\right).
\end{aligned}\end{equation}
To get larger $TV$, we hope $|Q_{1,4}|$ to be distributed around 0 following the distribution $P(\vec{y}|\vec{r})$, which requires at least two eigenvalues of $\Sigma=U(V_\rho+V_\Pi)U^T$ to approach zero. Note that
\begin{equation}
UV_\rho U^T=\left[\begin{matrix}
e^{-2s} & 0 & 0 & 0\\
0 & e^{2s} & 0 & 0\\
0 & 0 & e^{2s} & 0\\
0 & 0 & 0 & e^{-2s}
\end{matrix}\right].
\end{equation}
Following the argument in the proof of Theorem \ref{thm:I_GLOCC}, GLOCC cannot have $V_\Pi$ such that at least two eigenvalues of $\Sigma$ approach zero, $TV$ of GLOCC is constant with sufficient deviation from 1. But for nonlocal measurement it is possible to achieve $TV\rightarrow1$ using the measurement constructed in Eq.~\ref{nonlocal}.
\end{proof}

\begin{figure}[!tb]
\begin{center}
\includegraphics[width=0.4\textwidth]{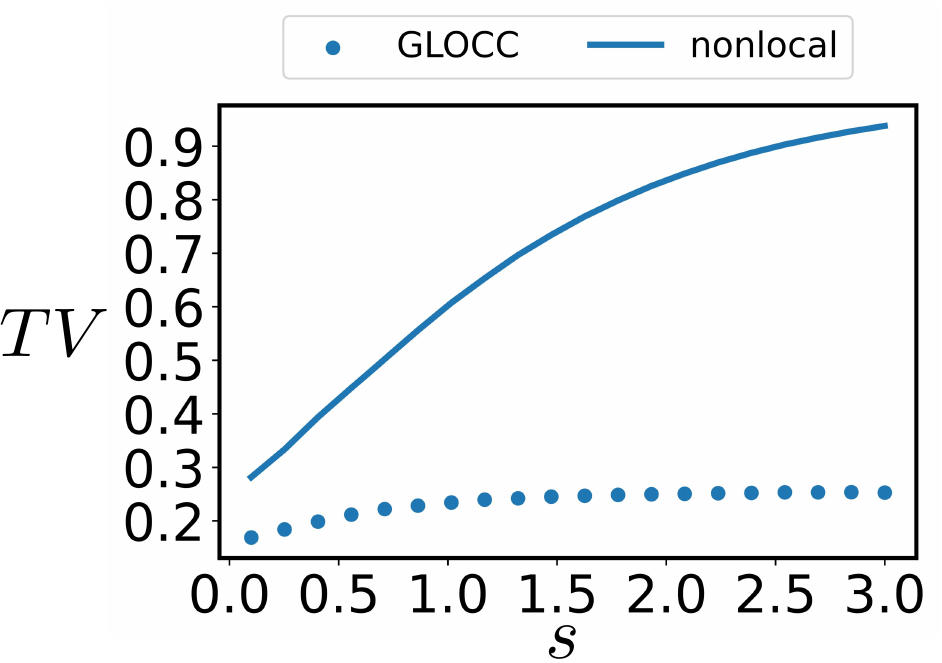}
\caption{The total variation distance $TV$ is presented as a function of the squeezing parameter $s$ for both GLOCC and nonlocal Gaussian measurements. For GLOCC, $TV$ is obtained through numerical optimization. In the case of nonlocal Gaussian measurements, $TV$ is calculated based on the measurements of $q_1 - q_2$ and $p_1 + p_2$, as introduced in Proposition \ref{prop:nonlocal_Gaussian}. Here, the parameter $\sigma$ is set to 1.} 
\label{fig:tv}
\end{center}
\end{figure}

Although the analytical proof shows that GLOCC cannot achieve $TV\rightarrow1$, we are unable to analytically determine a tight upper bound for $TV$ using GLOCC.
We now numerically optimize the GLOCC measurement to determine the optimal total variation $TV$. As shown in Fig.~\ref{fig:tv}, as the squeezing parameter $s$ increases, the $TV$ for the nonlocal Gaussian measurement approaches 1, whereas the $TV$ for the GLOCC measurement remains significantly below 1 as claimed in Proposition \ref{GLOCC_one}.

So far, we have demonstrated the gap in error probability when inferring the single bit of classical information encoded in the sign of $\alpha$ using GLOCC and nonlocal Gaussian measurement.
Similar to the approach of Ref.~\cite{terhal2001hiding}, we aim to further reduce the success probability of obtaining this bit of information  using GLOCC to nearly 1/2, equivalent to a random guess.  This can be accomplished by sending states with more modes 
to amplify the difference between the success probability. 
\begin{thm}
Consider the data hiding states, 
\begin{equation}
\begin{aligned}
\rho_{\text{even}}=2\int_{\text{even}}d\vec{r} \,P(\vec{r})\rho_{\vec{r}},\quad \rho_{\text{odd}}=2\int_{\text{odd}}d\vec{r} \,P(\vec{r})\rho_{\vec{r}},
\end{aligned}
\end{equation}
where prefactor 2 is introduced for normalization, $\vec{r}=[\vec{r}^1,\vec{r}^2,\cdots,\vec{r}^N]^T$, $\vec{r}^i=[a_i,b_i,c_i,d_i]^T$, $\rho_{\vec{r}}=\otimes_{i=1}^N\rho_{\vec{r}^i}$,  each $\rho_{\vec{r}^i}$ is given in Eq.~\ref{rho_two_mode}, $\int_{\text{even,odd}}$ denotes the integral over all the $\vec{r}$ such that the count of $\alpha_i= (a_i - c_i)(b_i + d_i) > 0$ is an even or odd number. And we assume prior distribution $P(\vec{r})$ is again Gaussian similar to Eq.~\ref{eq:prior_I} with $V_r=\sigma^2 I_{4N}$. The one bit of classical information is encoded in whether the state is $\rho_{\text{even}}$ or $\rho_{\text{odd}}$. The states $\rho_{\text{even},\text{odd}}$ can be distinguished with near 100\% success probability using nonlocal Gaussian measurements. In contrast, when restricted to GLOCC, the success probability for correctly distinguishing $\rho_{\text{even},\text{odd}}$ is close to $1/2$, equivalent to making a random guess.
   
\end{thm}
\begin{proof}
The covariance matrix of the state $\rho_{\vec{r}}$ is
\begin{equation}
V^N_\rho=\bigoplus_{i=1}^N V_\rho,
\end{equation}
where $V_{\rho}$ is given in Eq.~\ref{eq:TMSV}.
And the total variation between $P(\vec{r}|\text{even,odd})=\tr(\rho_{\text{even,odd}}\Pi_{\vec{r}})$ is given by
\begin{equation}\begin{aligned}\label{eq:TV_N_GLOCC}
&TV^N=\frac{1}{2}\int d\vec{y}\left|P(\vec{r}|\text{even})-P(\vec{r}|\text{odd})\right|\\
&=\int d\vec{y}\left|{\int_{\text{even}} d\vec{r}P(\vec{y}|\vec{r})P(\vec{r})}-{\int_{\text{odd}} d\vec{r}P(\vec{y}|\vec{r})P(\vec{r})}\right|\\
&=\int_{\text{even}} d\vec{y}\left({\int_{\text{even}} d\vec{r}P(\vec{y}|\vec{r})P(\vec{r})}-{\int_{\text{odd}} d\vec{r}P(\vec{y}|\vec{r})P(\vec{r})}\right)\\
&\,+\int_{\text{odd}} d\vec{y}\left({\int_{\text{odd}} d\vec{r}P(\vec{y}|\vec{r})P(\vec{r})}-{\int_{\text{even}} d\vec{r}P(\vec{y}|\vec{r})P(\vec{r})}\right)\\
&=TV^N(e,e)-TV^N(o,e)+TV^N(o,o)-TV^N(e,o),
\end{aligned}\end{equation}
where $\vec{y}=[\vec{y}^1,\vec{y}^2,\cdots,\vec{y}^N]^T$, $\vec{y}^i=[y^i_1,y^i_2,y^i_3,y^i_4]^T$, the even or odd labeled for the integral interval means the number of $\alpha_i=(a_i-c_i)(b_i+d_i)$ or $\beta_i=(y^i_1-y^i_3)(y^i_2+y^i_4)$ is even or odd. $TV(e,e)$ represents the count of $\alpha_i > 0$ or $\beta_i>0$ is an even number, with similar definitions for $TV(e,o)$, $TV(o,e)$, and $TV(o,o)$. We again change variables and switch the order of integral in a similar fashion, which allows us to simplify the equations as
\begin{equation}
\begin{aligned}
TV^N(e,e) &= \sum_{\substack{\#\,\alpha_i>0\,\text{even}\\\#\,\beta_i>0\,\text{even}}} 
\left(\prod_{i=1}^N \int_{\beta_i}\!\!d\vec{y}^i \int_{\alpha_i}\!\!d\vec{r}^i\, P(\vec{r}^i)\right) P(\vec{y}|\vec{r}) \\
&= \sum_{\substack{\#\,\alpha_i>0\,\text{even}\\\#\,\beta_i>0\,\text{even}}}
\prod_{i=1}^N \hat{f}(\alpha_i, \beta_i) P(\vec{y}|\vec{r}).
\end{aligned}
\end{equation}
where $\sum_{\substack{\#\,\alpha_i>0\,\text{even}}} $ denotes a summation over all cases where the count of $\alpha_i > 0$ is an even number.   We emphasize that $\hat{f}(\alpha_i, \beta_i)$ is not a function as it includes an integral involving $P(\vec{y} | \vec{r})$. 
\begin{equation}\begin{aligned}
&TV^N(e,e)-TV^N(e,o)\\
&=\left(\sum_{\substack{\#\,\alpha_i>0\,\text{even}\\\#\,\beta_i>0\,\text{even}}} -\sum_{\substack{\#\,\alpha_i>0\,\text{odd}\\\#\,\beta_i>0\,\text{even}}} \right)
\left(\prod_{i=1}^N \hat{f}(\alpha_i,\beta_i)\right) P(\vec{y}|\vec{r})\\
&=\sum_{\#\,\beta_i>0\,\text{even}}\prod_{i=1}^N\left(-\hat{f}(\alpha_i\!>\!0,\beta_i)+\hat{f}(\alpha_i\!<\!0,\beta_i)\right) P(\vec{y}|\vec{r}), 
\end{aligned}\end{equation}
\begin{equation}\begin{aligned}
&TV^N(o,o)-TV^N(o,e)\\
&=\left(\sum_{\substack{\#\,\alpha_i>0\,\text{odd}\\\#\,\beta_i>0\,\text{odd}}} -\sum_{\substack{\#\,\alpha_i>0\,\text{even}\\\#\,\beta_i>0\,\text{odd}}} \right)
\left(\prod_{i=1}^N \hat{f}(\alpha_i,\beta_i)\right) P(\vec{y}|\vec{r})\\
&=-\sum_{\#\,\beta_i>0\,\text{odd}}\prod_{i=1}^N\left(-\hat{f}(\alpha_i\!>\!0,\beta_i)+\hat{f}(\alpha_i\!<\!0,\beta_i)\right) P(\vec{y}|\vec{r}), 
\end{aligned}\end{equation}
\begin{equation}\begin{aligned}
TV^N&=\prod_{i=1}^N\bigg[\hat{f}(\alpha_i>0,\beta_i>0)-\hat{f}(\alpha_i>0,\beta_i<0)\\
&\quad-\hat{f}(\alpha_i<0,\beta_i>0)+\hat{f}(\alpha_i<0,\beta_i<0)\bigg]P(\vec{y}|\vec{r})\\
&=\prod_{i=1}^N\Bigg[\int_{-\infty}^{+\infty}dQ^i_1\int_{-\infty}^{+\infty}dQ^i_2\int_{-\infty}^{+\infty}dQ^i_3\int_{-\infty}^{+\infty}dQ^i_4 \\
&\quad\times\left[1-\text{erf}\left(\frac{|Q_1^i|}{\sqrt{2}\sigma}\right)\right]\left[1-\text{erf}\left(\frac{|Q_4^i|}{\sqrt{2}\sigma}\right)\right]\Bigg]P(\vec{y}|\vec{r}).
\end{aligned}\end{equation}
For each $i$, we need two eigenvalues approaching zero to have $TV^N\rightarrow1$, which cannot be achieved using GLOCC. Because for any of the $i$th mode, if we have $|Q_{1,4}^i|\rightarrow 0$, this requires both $\omega_{\rho3}^i=[0,\cdots,0,\omega_{\rho3},0,\cdots,0]$ and $\omega_{\rho4}^i=[0,\cdots,0,\omega_{\rho4},0,\cdots,0]$ to be  eigenvectors of $V^N=V_{\rho}^N+V_{\Pi}^N$ with eigenvalues approaching zero, where $\omega_{\rho3,4}$ is given in Eq.~\ref{Vrho_spectrum}. Since $\omega_{\rho3,4}$ are already the eigenvectors of $V_\rho^N$ with eigenvalues $e^{-2s}$, this means $\omega_{\rho3,4}$ must also be the eigenvectors of $V_\Pi^N$ with vanishing eigenvalues. 
We now want to check whether $T_NV_\Pi^N T_N+i \Omega_N\geq 0$, where $\Omega_N=\bigoplus_{i=1}^N\Omega$, $T_N=\bigoplus_{i=1}^N T$. Since $\omega_{\rho3,4}^i$ are eigenvectors of $V_\Pi^N$ with vanishing eigenvalues, $T_N\omega_{\rho3,4}^i$  are eigenvectors of $T_NV_\Pi^N T_N$ with vanishing eigenvalues. Note that $T_N\omega_{\rho3}^i+iT_N\omega_{\rho4}^i$ is an eigenvectors of $i\Omega_N$ with eigenvalues $-1$. This means $T_NV_\Pi^N T_N+i \Omega_N$ is not positive-semidefinite and hence the measurement is not PPT. For each $i$, $TV^N$ of GLOCC will get a factor deviating from 1 due to the integral of $Q_{1,4}^i$, which means $TV^N\rightarrow 0$ as $N\rightarrow \infty$. But nonlocal measurement can have $TV^N\rightarrow 1$, which can be achieved by implementing the measurement in Eq.~\ref{nonlocal} for each $\rho_{\vec{r}_i}$. We have thus constructed an example of data hiding based on the displaced two mode squeezed states from GLOCC. 

    
\end{proof}

\section{Data hiding from general Gaussian operations}

In this section, we consider data hiding from general Gaussian operations. We begin by considering a two-mode weak thermal state with zero displacement, described by the P representation \cite{mandel1995optical} 
\begin{equation}\begin{aligned}\label{rho_thermal}
&\rho=\int\frac{d^2\alpha d^2\beta}{\pi^2\det\Gamma}\exp(-\vec{\gamma}^\dagger \Gamma^{-1}\vec{\gamma})\ket{\vec{\gamma}}\bra{\vec{\gamma}},\\
&\vec{\gamma}=[\alpha,\beta]^T,\quad \Gamma=  \frac{\epsilon}{2}\left[
\begin{matrix}
1 & |g|e^{i\theta}\\
|g|e^{-i\theta} & 1
\end{matrix}\right],\\
&\ket{\vec{\gamma}}=\exp(\alpha \hat{a}^\dagger-\alpha^*\hat{a})\exp(\beta \hat{b}^\dagger-\beta^*\hat{b})\ket{0},
\end{aligned}\end{equation}
where $\epsilon$ is the mean photon number per temporal mode and assumed to be much less than one $\epsilon\ll 1$, $\hat{a},\hat{b}$ are the annihilation operators for the two modes.
The motivation for selecting this state comes from Ref.~\cite{tsang2011quantum}, which shows that nonlocal schemes for estimating $\theta$ outperform any LOCC, highlighting the advantage of nonlocal measurements. 
This inspired us to encode information in the phase $\theta$ of a two-mode weak thermal state for quantum data hiding, effectively concealing the information from LOCC. 
However, while this was our initial motivation, we made an unexpected discovery: Gaussian measurements, including nonlocal ones, perform worse than proper nonlocal non-Gaussian measurements. This reveals a scheme to hide one classical bit of information from any Gaussian measurement using only a separable state.

To hide one bit of classical information, we choose $|g|=1$ and $\theta=0,\pi$ in Eq.~\ref{rho_thermal} as the two states $\rho_{\pm}$, which have the following covariance matrices
\begin{equation}
V_\pm=\left[
\begin{matrix}
1+\epsilon & 0 & \pm\epsilon & 0\\    
0 & 1+\epsilon & 0 & \pm\epsilon\\
\pm\epsilon & 0 & 1+\epsilon & 0\\
0 & \pm\epsilon & 0 & 1+\epsilon
\end{matrix}\right].
\end{equation}
For non-Gaussian measurement, it is easier to expand Eq.~\ref{rho_thermal} in Fock basis as a series of $\epsilon$,
\begin{equation}
\rho=(1-\epsilon)\ket{00}\bra{00}+\frac{\epsilon}{2}\left[\begin{matrix}
1 & |g|e^{i\theta}\\
|g|e^{-i\theta} & 1
\end{matrix}\right]+o(\epsilon).
\end{equation}
\begin{proposition}
Consider the pair of data hiding states $\rho_{+}^{\otimes N},\rho_{-}^{\otimes N}$, which is used to encode one bit of classical information. There exists a non-Gaussian measurement that can distinguish $\rho_{+}^{\otimes N},\rho_{-}^{\otimes N}$ with success probability close to 100\% when $N=\Theta(1/\epsilon)$.
\end{proposition}
\begin{proof}

We first consider the success probability of a nonlocal non-Gaussian measurement on one copy of state $\rho$ in Eq.~\ref{rho_thermal}. The projective measurement and the corresponding probability of obtaining each outcome is 
\begin{equation}\begin{aligned}\label{nonGaussian_POVM}
&\ket{00}\bra{00},\quad P=1-\epsilon,\\
&\ket{\pm}\bra{\pm},\quad P=\frac{\epsilon}{2}(1\pm |g|\cos\theta),    
\end{aligned}
\end{equation}
where $\ket{\pm}=(\ket{01}\pm\ket{10})/\sqrt{2}$, $\ket{0},\ket{1}$ are the vacuum and single photon states. We do the same measurement for $N$ copies of the same state $\rho^{\otimes N}$, which has the probability
\begin{equation}\begin{aligned}
&P\big(k,m\big||g|,\theta\big)=C_N^kC_{N-k}^m(1-\epsilon)^k\\
&\times\left(\frac{\epsilon}{2}(1+|g|\cos\theta)\right)^m\left(\frac{\epsilon}{2}(1-|g|\cos\theta)\right)^{N-m-k},
\end{aligned}
\end{equation}
where $k=0,1,2,\cdots,N$ labels the number of times of getting $\ket{00}$, $m=0,1,2,\cdots,N-k$ labels the number of times we get $\ket{+}$. The total variation between the probability of measuring the two data hiding states $\rho_+^{\otimes N}$ and $\rho_-^{\otimes N}$ is
\begin{equation}\begin{aligned}
TV&=\frac{1}{2}\sum_{k=0}^N\sum_{m=0}^{N-k}|P(k,m|+)-P(k,m|-)|\\
&=\frac{1}{2}\sum_{k=0}^N\sum_{m=0}^{N-k}C_N^kC_{N-k}^m(1-\epsilon)^k\\
&\quad\quad\quad\quad\times|\epsilon^m 0^{N-k-m}-0^m \epsilon^{N-k-m}|\\
&=\sum_{k+m=N,m\neq 0}C_N^k(1-\epsilon)^k\epsilon^m\\
&=1-(1-\epsilon)^N\approx N\epsilon+o(\epsilon).
\end{aligned}\end{equation}
So, to achieve a success probability close to one, we need $N=\Theta(1/\epsilon)$.
\end{proof}

\begin{thm}
Consider the pair of data hiding states $\rho_{+}^{\otimes N},\rho_{-}^{\otimes N}$, which is used to encode one bit of classical information. Any Gaussian measurement can distinguish  $\rho_{+}^{\otimes N},\rho_{-}^{\otimes N}$ with error probability $P_{\text{err}}\geq \frac{1}{2}(1-\sqrt{2N}\epsilon)$ for $N=o(\epsilon^{-2})$. 
When $N = \Theta(\epsilon^{-1})$ and $\epsilon \to 0$, any Gaussian measurement performs almost no better than random guessing.
\end{thm}
\begin{proof}
For any Gaussian measurement, the total variation distance between two probability $P(\vec{y}|\pm)=\tr(\rho_{\pm}^{\otimes N}\Pi_{\vec{y}})$ is 
\begin{equation}\begin{aligned}
&TV=\frac{1}{2}\int d\vec{y}|P(\vec{y}|+)-P(\vec{y}|-)|,\\
&P(\vec{y}|\pm)=\frac{1}{(2\pi)^2\sqrt{\det V_N}}\exp(-\frac{1}{2}\vec{y}^TV_N^{-1}\vec{y}),
\end{aligned}
\end{equation}
where $V_N=V_{\pm}\otimes I_N+V_\Pi$, $V_\Pi$ is the covariance matrix describing the POVM of a Gaussian measurement, $\vec{y}$ consists of $\vec{y^i}=[y^i_1,y^i_2,y^i_3,y^i_4]^T$ as in Eq.~\ref{eq:TV_N_GLOCC}. We again define $\vec{Y^i}=U\vec{y^i}$ as in Eq.~\ref{eq:U}
and get
\begin{equation}\begin{aligned}
&P(\vec{y}|\pm)=P(\vec{Y}|\pm)=\frac{1}{(2\pi)^2\sqrt{\det \Sigma_N}}\exp(-\frac{1}{2}\vec{Y}^T\Sigma_N^{-1}\vec{Y}),\\
&\Sigma_N=\Sigma_{\pm}+\Sigma_\Pi, \quad \Sigma_{\pm}=(UV_{\pm}U^T)\otimes I_N,\\
&\Sigma_\Pi=(U\otimes I_N)V_\Pi (U\otimes I_N)^T,\\
&UV_+U^T=\left[\begin{matrix}
1 & 0 & 0 & 0 \\
0 & 1+2\epsilon & 0 & 0\\
0 & 0 & 1 & 0\\
0 & 0 & 0 & 1+2\epsilon
\end{matrix}\right],\\
&UV_-U^T=\left[\begin{matrix}
1+2\epsilon & 0 & 0 & 0 \\
0 & 1 & 0 & 0\\
0 & 0 & 1+2\epsilon & 0\\
0 & 0 & 0 & 1
\end{matrix}\right],
\end{aligned}\end{equation}
where $I_N$ is the identity matrix added for the direct product of $N$ copies of states. The total variation can be upper bounded by Pinsker's inequality \cite{csiszar2011information,pardo2018statistical}
\begin{equation}
TV(P,Q)\leq\sqrt{\frac{1}{2}D_{KL}(P||Q)},
\end{equation}
where $D_{KL}(P||Q)$ is Kullback–Leibler divergence for two probability distribution $P,Q$. For the Gaussian probability distribution $P=\mathcal{N}(\mu_1,\Sigma_1)$, $Q=\mathcal{N}(\mu_2,\Sigma_2)$ \cite{csiszar2011information,pardo2018statistical}
\begin{equation}\begin{aligned}
&D_{KL}(\mathcal{N}(\mu_1,\Sigma_1)||\mathcal{N}(\mu_2,\Sigma_2))
\\
&=\frac{1}{2}\bigg[\tr(\Sigma_2^{-1}\Sigma_1)+(\mu_2-\mu_1)^T\Sigma_2^{-1}(\mu_2-\mu_1)\\
&\quad\quad\quad\quad\quad\quad-M+\ln\left(\frac{\det\Sigma_2}{\det\Sigma_1}\right)\bigg],
\end{aligned}\end{equation}
where $M$ is the dimension of the matrix $\Sigma_{1,2}$. For our case
\begin{equation}\begin{aligned}
&\Sigma_1=\Sigma_++\Sigma_\Pi=A+\epsilon B_1,\\
&\Sigma_2=\Sigma_-+\Sigma_\Pi=A+\epsilon B_2,\\
&A=I_{4N}+\Sigma_\Pi,\quad\mu_{1,2}=\vec{0},\quad M=4N,\\
&B_1=\left[\begin{matrix}
0 & 0 & 0 & 0\\
0 & 2I_{N} & 0 & 0\\
0 & 0 & 0 & 0\\
0 & 0 & 0 & 2I_{N} \\
\end{matrix}\right],\\
&B_2=\left[\begin{matrix}
2I_{N} & 0 & 0 & 0\\
0 & 0 & 0 & 0\\
0 & 0 & 2I_{N} & 0\\
0 & 0 & 0 & 0 \\
\end{matrix}\right].
\end{aligned}
\end{equation}
For $\Sigma=A+\epsilon B$, we can have the following expansion \cite{petersen2008matrix}
\begin{equation}\begin{aligned}
\Sigma_2^{-1}\Sigma_1&=I+\epsilon A^{-1}(B_1-B_2)\\
&+\epsilon^2 A^{-1}B_2 A^{-1}(B_2-B_1)+o(\epsilon^2),
\end{aligned}\end{equation}
\begin{equation}\begin{aligned}
&\ln\left(\frac{\det\Sigma_2}{\det\Sigma_1}\right)=\epsilon\tr[A^{-1}(B_2-B_1)]\\
&+\epsilon^2\tr(A^{-1}B_1 A^{-1}B_1-A^{-1}B_2 A^{-1}B_2)+o(\epsilon^2),
\end{aligned}\end{equation}
\begin{equation}
D_{KL}=\frac{\epsilon^2}{4}\tr(A^{-1}(B_1-B_2)A^{-1}(B_1-B_2)).
\end{equation}
So, now the problem is the maximization of $\tr(A^{-1}(B_1-B_2)A^{-1}(B_1-B_2))$, which can be easily bounded by
\begin{equation}
\tr(A^{-1}(B_1-B_2)A^{-1}(B_1-B_2))\leq\tr(A^{-2}).
\end{equation}
Note that $A=I_{4N}+\Sigma_\Pi$, where $\Sigma_\Pi\geq 0$, we have
\begin{equation}
0\leq A^{-1}\leq I.
\end{equation}
We then have
\begin{equation}
D_{KL}\leq 4\epsilon^2N,\quad TV\leq \sqrt{2N}\epsilon.
\end{equation}
The required  $N$ to achieve nearly 100\% success probability of distinguishing $\rho_{\pm}^{\otimes N}$ is at least $\Omega(1/\epsilon^2)$. 
\end{proof}

If we choose $\epsilon \rightarrow 0$ and $N =\Theta( 1/\epsilon)$, the data hiding states $\rho_+^{\otimes N}$ and $\rho_-^{\otimes N}$ cannot be distinguished by Gaussian measurements but can be distinguished by the non-Gaussian measurement described in Eq.~\ref{nonGaussian_POVM}. We have thus constructed the data hiding states for general Gaussian operations. Let's check this discussion with some specific examples. Firstly, let's consider the local heterodyne detection at each mode, which has $\Sigma_\Pi=I_{4N}$, and 
\begin{equation}
D_{KL,hetero}=N\epsilon^2,
\end{equation}
which shows we need at least $N=\Omega(\epsilon^{-2})$ to have nearly 100\% success probability for distinguishing $\rho_{\pm}^{\otimes N}$.

Secondly, let's consider the case when we first combine the light from $\hat{a},\hat{b}$ modes on a balanced beam splitter and do homodyne detection on the two output ports. This measurement is described by
\begin{equation}
\Sigma_\Pi=\lim_{s\rightarrow\infty}\left[\begin{matrix}
e^{-2s}I_{N} & 0 & 0 & 0\\
0 & e^{2s}I_{N} & 0 & 0\\
0 & 0 & e^{2s}I_N & 0\\
0 & 0 & 0 & e^{-2s}I_N
\end{matrix}\right],
\end{equation}
\begin{equation}
D_{KL}=2N\epsilon^2.
\end{equation}
Note that for the second case, the measurement is already a nonlocal measurement. Intuitively, we are projecting onto  displaced two mode squeezed states with infinite squeezing. And with such a nonlocal measurement, we still at least need $N=\Omega( 1/\epsilon^2)$ to have nearly 100\% success probability for distinguishing $\rho_{\pm}^{\otimes N}$ as claimed.

In the above two example, we do not yet saturate the upper bound for $D_{KL}$. This is because, we only use the fact that $\Sigma_\Pi\geq 0$ in the derivation of upper bound. But as a valid POVM, $\Sigma_\Pi$ also need to satisfy some condition from the uncertainty principle \cite{weedbrook2012gaussian}. Indeed, we can easily read from the proof that, the following $\Sigma_\Pi$ achieves the upper bound
\begin{equation}
\Sigma_\Pi=\lim_{s\rightarrow\infty}e^{-2s}I_{4N}.
\end{equation}
This measurement essentially asks for the perfect estimation of all quadratures at the same time, which is not physically allowed.


\section{Conclusion}

In this paper, we investigate two distinct scenarios of data hiding within the CV context. First, we introduce data hiding with respect to a new  class of operations, GLOCC. To establish the  intuition, we identify the CV counterpart to the challenge of distinguishing Bell entangled states using LOCC, which is a key insight underlying the initial data hiding proposal \cite{terhal2001hiding,divincenzo2002quantum}. Building on this intuition, we construct data-hiding states under GLOCC using the mixture of displaced two-mode squeezed states.

Second, we explore data hiding against general Gaussian operations. We propose a novel example of data-hiding states utilizing multiple copies of two-mode thermal states in the weak-strength limit. Notably, this construction is closely related to the interferometric imaging described in Ref.~\cite{tsang2011quantum,gottesman2012longer}. As such, the demonstrated advantage of non-Gaussian operations for data hiding may also suggest a broader advantage of non-Gaussian operations in imaging applications.



Several questions remain unanswered. Our data hiding states from GLOCC are mixtures of displaced two-mode squeezed states, which are non-Gaussian. It is intriguing to consider whether we can construct data hiding states from GLOCC using only Gaussian states. Additionally, our current data hiding scheme only hides one bit of classical information. Given that we are working with CV states, it is worth exploring whether we can hide a classical continuous variable.

\section*{Acknowledgements}
We thank Debbie Leung, Peixue Wu, Mao Lin for helpful discussions.
Y.W. and G.S. acknowledges funding from the Canada First Research Excellence Fund. Y.W. also acknowledges funding provided by Perimeter Institute for Theoretical Physics, a research institute supported in part by the Government of Canada through the Department of Innovation, Science and Economic Development Canada and by the Province of Ontario through the Ministry of Colleges and Universities.

\bibliography{main}

\end{document}